\documentclass{ws-ijfcs}
\usepackage{enumerate}
\usepackage{url}
\usepackage{graphicx}
\usepackage{subcaption}
\usepackage{pgfplots}
\usepackage{float}
\usepackage{textcomp}
\usepackage{xcolor}
\usepackage{cite}
\usepackage{verbatim}
\usepackage{algorithm}
\usepackage{tikz}
\usepackage{algpseudocode}

\begin{document}

\title{Approximation Algorithms for the UAV Path Planning with Object Coverage Constraints}

\author{Jiawei Wang, Vincent Chau, Weiwei Wu\thanks{W. Wang, V. Chau, W. Wu are with the Department of Computer Science and Engineering, Southeast University, Nanjing, China. E-mail: \{220211985, vincentchau, weiweiwu\}@seu.edu.cn. This work is supported by the National Natural Science Foundation of China under Grant No. 62202100, 62232004 and the Natural Science Foundation of Jiangsu Province under Grant No.  BK20230024.}
}

\maketitle

\begin{abstract}
We study the problem of the Unmanned Aerial Vehicle (UAV) such that a specific set of objects needs to be observed while ensuring a quality of observation. Our goal is to determine the shortest path for the UAV. This paper proposes an offline algorithm with an approximation of $(2+2n)(1+\epsilon)$ where $\epsilon >0$ is a small constant, and $n$ is the number of objects. We then propose several online algorithms in which objects are discovered during the process.
To evaluate the performance of these algorithms, we conduct experimental comparisons. Our results show that the online algorithms perform similarly to the offline algorithm, but with significantly faster execution times ranging from 0.01 seconds to 200 seconds. We also show that our methods yield solutions with costs comparable to those obtained by the Gurobi optimizer that requires 30000 seconds of runtime.
\end{abstract}


\section{Introduction}
\label{intro}

Owing to the recent rapid development in microchips and various sensors, the growing popularization of Unmanned Aerial Vehicles (UAVs) is reflected in its extensive applications, such as structure inspection, smart farming, wildfire detection, cinematography, etc. As an intelligent and integrated platform, UAVs are capable of accomplishing tasks that are difficult or dangerous for human. For instance, an UAV can fly around a building and reconstruct it finely from the photos it took by following a planned path, which is hard for human operators~\cite{uav_structureinspection}. Being equipped with specific device, UAVs liberate farmers from laborious work by irrigating farmland automatically over a large area~\cite{smart_farm}. Moreover, UAVs can perform missions in disastrous areas~\cite{fire_detect} and elevate the photography effect to new height via approaching somewhere not available for human~\cite{uav_cinema}. 

Benefiting from its flexibility and maneuverability in practice, a UAV generally requires a trajectory in order to accomplish various tasks in which they fly. Correspondingly, coverage path planning aims to find a feasible path for UAV that covers an area or a set of Points of Interests (PoI) if satisfying some conditions, which can be a distance constraint or hardware parameters of device such as the camera angle. Distance constraint is that the UAV has a perception range with radius supported by cameras or radars and the target object should be within this distance from the UAV. In order to make a complete coverage of a certain space, a wide range of algorithms were proposed in literature, planning feasible paths so as to cover PoIs or to cover the space by summing up the area swept over during its flight. 

In this work, we focus on a realistic situation in which a UAV is expected to cover a set of objects. Besides, every side of the object has to be observed, implying that it cannot be treated merely as a point cover problem. Such a problem can be widely extended to building inspection in which a path guides the UAV to take photos of each side of the buildings. The offline and online algorithms proposed will reduce the need for human intervention and enhance the efficiency and thoroughness of the coverage task, leading to more systematic and dependable data gathering. 

\section{Related Work} 
\label{related-work}

Given that our work plans a path for a UAV in a coverage problem, we present the related works in three parts. The first part is devoted to different methods of UAV path planning. Then the focus is transferred to UAV coverage problem, in which the main task is to find a path that achieves effective coverage. Finally, we briefly introduce several algorithms of the classical Travelling Salesman Problem (TSP). 

\subsection{UAV path planning}
UAV path planning is essential for the control of unmanned aerial vehicles, which are able to fly automatically under the planned path. There are numerous methods proposed in the literature to devise efficient paths for UAVs and they can be classified into the following categories: sampling based algorithms, mathematical model based algorithms, bio-inspired algorithms and artificial intelligence algorithms. Sampling based algorithms require the information of the map to be known in advance and sample or divide the map into a set of nodes. The concrete technique can be represented by rapid-exploring random trees (RRT), RRT-star, A-star. Mathematical methods mainly include Lyapunov function~\cite{Lyapunov}, which is used to maintain the stability of UAVs, linear programming, which integrate all the cost factors with Hamiltonian function to search for an optimal path, as described in~\cite{math-planning}, and Beziere curve which further consider the smoothness of the flight path. Bio-inspired technique~\cite{bio-planning} adopts evolutionary ideas by selecting a path as a parent path and generating new path via mutation and crossover. Selection process is led by the adaptive value of offspring. AI-based methods are currently studied greatly. Traditional machine learning models, such as k-means, used to determine the path in complex multi-task environment via clustering the target points~\cite{k-means-planning}. Neural network enables the UAV to perceive the environment by vision and thus solve the navigation and position problems~\cite{NN-planning1},~\cite{NN-planning2}. Apart from neural network, reinforcement learning is used extensively in path planning where Q-learning is carried out in~\cite{q-learning-planning} and G-learning is proposed in~\cite{g-learning-planning}. In order to further enhance the learning efficiency of these two methods, Lei et al.\cite{DRL} modeled the UAV navigation problem as a Markov decision process and proposed a model interpretation method based on feature attributes to explain the behavior of UAV in the process of path planning. Similarly, Xie et al.\cite{DRL2} expresses path planning as a partially observable Markov process, constructs RNN to solve partially observable problems, and uses reward value and action value to reduce meaningless exploration.

\subsection{UAV coverage}
As a sub-problem of UAV path planning, UAV coverage additionally requires UAVs to achieve a coverage task, as the name indicated. One type is regional coverage, which assumes that a region is covered when it is in the perception area of the UAV. A relatively conventional algorithm is presented in~\cite{cellular} where Nam et al. used approximate cellular decomposition with criteria on both length and the number of turns on the path. Yao et al.\cite{river_rescue} proposed an offline path planning method based on Gaussian mixual model and heuristic prioritization, which is aimed at maximizing the probability of finding the lost target in a river rescue mission. Xie et al.~\cite{multi_polygon} solves the problem of multiple polygon regions by integrating covering a single polygon and traveling salesman problem. These are the methods for one UAV and algorithms built on multiple UAVs are as follows. Jing et al.~\cite{uav_structureinspection} navigate the UAV to inspect a large complex building structure by combining set covering problem and vehicle routing problem. Rapidly exploring random tree, as indicated by its name, is broadly used to explore an unknown region with high efficiency. In~\cite{surveillance}, a cooperative surveillance task was achieved by multiple UAVs where RRT was modified to find feasible trajectories passing suitable observing locations, on which particle swarm optimization is then performed. Focusing on the coverage path planning of heterogeneous UAVs, the authors in \cite{hetero-cover} established the UAV and regional model before using linear programming to accurately provide the best point-to-point flight path for each UAV. Then, inspired by the foraging behavior of ants, a heuristic algorithm based on ACS is proposed to search the approximate optimal solution and minimize the time consumption of tasks in the cooperative search system.

The other type is object-oriented, where the problem aims to cover a set of objects and the condition satisfied when the distance between the UAV and an object is less than a threshold. In this case, effective coverage can be treated as a binary in that 1 for covered and 0 for not covered and this can find applications in the context of wireless sensor networks. Huang et al.\cite{seuwsn} embedded turning angles and switching numbers during retrieving data into graph structure and obtained a path for a UAV via generalized TSP solutions. Gong et al.~\cite{straightwsn} completed the data collection mission in a straight line situation which minimized the UAV's total flight time via dynamic programming while retrieving a certain amount of data from each sensor. After knowing the possible flying path, Yang et al.\cite{bio-inspired} combined a genetic algorithm and ant colony optimization to select the best path for data collection. In~\cite{robotcover}, the agents regard moving objects and obstacles as disks of different sizes and the goal is to find a collision-free coverage path in a dynamic changing environment while ensuring path smoothness. As for multiple UAVs, Alejo et al.~\cite{wsn} integrated online RRT with genetic algorithms, guiding several UAVs collecting data simultaneously from sensors randomly distributed. A multi-agent architecture is designed in~\cite{multi-agent-catastrophe} for UAVs to patrol in a region and monitor key ground facilities.

\subsection{Travelling Salesman Problem (TSP)}
The Travelling Salesman Problem and its variants are thoroughly studied in literature. The problem is to find a path that starts and ends at the same place, such that a set of places need to be visited once. Authors in~\cite{DFJ-tsp,MTZ-tsp,GB-tsp,GP-tsp} designed different forms of integer linear programming for the problem. Christofides~\cite{1.5TSP} obtained a 3/2 approximation by combining a minimum spanning tree of the original graph and the best matching of the vertices with odd degree. Nearest neighbor algorithm~\cite{nearest-tsp} constructs the path by continuously adding the new vertex closest to the current one. Cheapest Insertion~\cite{cheapest-insertion} enlarges the path via inserting new node to the path with the lowest insertion cost and its approximation ratio is 2. Holland~\cite{genetic-tsp} employed genetic algorithms by fostering path offspring with the best adaptive score. Note that the TSP is an NP-hard problem, therefore we are interested in finding approximate solutions in polynomial time.

\subsection*{Contributions}

The literature summarized above admittedly achieved satisfying outcomes in different circumstances, yet few of them truly yields theoretical results suggesting the degree to which they are efficient relative to the optimal solutions. In our work, however, a result with an approximation ratio is obtained.
\\Our contributions are summarized as follows:
\begin{enumerate}
\item We design an Integer Linear Programming (ILP) that is $(1+\varepsilon)$-approximate with $\varepsilon>0$.
\item An offline method is proposed which achieves an approximation ratio of $(2+2n)(1+\epsilon)$ in theory and of around 2.1 in experiments.
\item Based on the offline method, we propose three online heuristics.
\item We carry out various numerical experiments and the results show that our methods can achieve close performance to that of the ILP solver at much faster speed.
\end{enumerate}

The structure of the rest of the paper is organized as follows. %We will make a summary of the related works from three parts in the rest of Section \ref{related-work}. 
In Section \ref{system-model}, the system model is presented followed by the formalization of the problem. The detail of our methods is the described in Section \ref{methodology}, including all four methods and the design of integer linear programming. The process and results of the experiments are presented in Section \ref{experiments}.

\section{System Model and Problem Formulation}
\label{system-model}
In this section, we first present the description of the system model, followed by the formalization of the goal of the problem, which aims to find a path for the UAV in both offline and online scenarios that minimizes the total length of the flying path while ensuring that all objects have been observed. After that, we describe the process to generate the observation points via area discretization and present the integer linear programming form of the problem at last.

\subsection{System Model}

Consider a two-dimensional space in which there are $n$ objects and a UAV equipped with a camera that can move freely in this space. The objects set is referred to as ${O}= \{o_1,o_2,o_3...o_n\}$. The location and the scale of each object $o_i$, which is rectangular, are described by $(x_{o_i}, y_{o_i})$, $(l_{o_i}, w_{o_i})$ where $x_{o_i}, y_{o_i}$ are the objects' coordinates and $l_{o_i}, w_{o_i}$ are respectively the length and width. We further define the midpoint of a side of each object to be $q_r$ where $r \in [1, 4n]$. The set of all the midpoints is $Q$. In the online setting, the UAV has a perception range and will gradually detect objects during its flight. Since the volumes of the objects are not negligible, all the four sides of an object have to be covered. Besides, the UAV cannot observe the side behind it because it is blocked by the front side. In order to ensure the observation quality of each side, the UAV is not allowed to take photos from an over-deviated orientation, which is supported by the fact that the most information is obtained when the camera is facing directly against an object. Moreover, since the distance also affects the observing quality, the UAV cannot observe the object when the distance exceeds a threshold. Concretely, let $a$, $b$ be the two endpoints of the side, $p$ the location of the UAV, and $t$ the normal vector of the side. As illustrated in Figure~\ref{img:efficient_observation}, the UAV can efficiently observe a side of an object if and only if it satisfies Definition~\ref{def:efficient_observe}. We adapt the definition to our problem, which is based on ~\cite{deploy} in order to emphasize the quality of observation.

\begin{definition}[Efficient observation]\label{def:efficient_observe} 
Let $a$, $b$ be the end points of a side of an object, $\vec{t}$ be the normal vector of the side, then the side is efficiently observed by the UAV at an observation point $p_i$ if and only if $\alpha(\vec{t}, \overrightarrow{p_i a}) \leq \theta$, $\alpha(\vec{t}, \overrightarrow{p_i b}) \leq \theta$, $distance \;||p_i a|| \leq d_{max} $ and $distance \;||p_i b|| \leq d_{max} $ where $\theta$ is the maximum observation angle, $d_{max}$ the observation range, and $\alpha(,)$ the angle between two vectors.
\end{definition}

\begin{figure}[htbp]
\centerline{\includegraphics[width=0.4\textwidth]{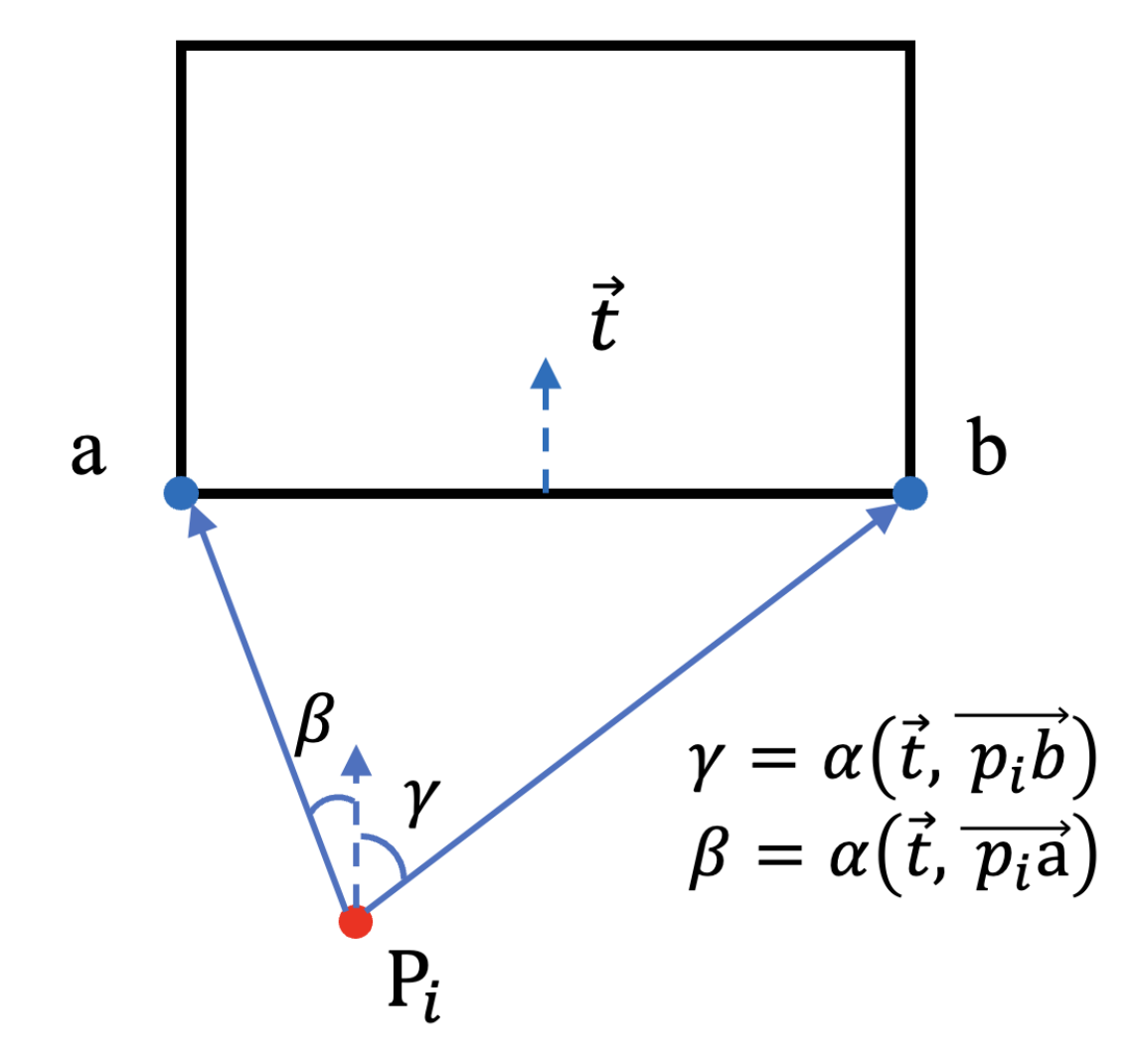}}
\caption{Efficient observation that satisfies both the distance and the angle constraints.}
\label{img:efficient_observation}
\end{figure}

Given that there is only one UAV expected to cover all these objects, the task is to minimize the path length while satisfying the coverage constraint, therefore it cannot be solved directly by a TSP algorithm which does not take into account the characteristics of the observation points. There exist some cases where the observation area of several objects overlaps, hence at some points the UAV is able to observe multiple objects by rotating its camera. On account of the fact that the layout of the objects can be any form, it is challenging to find a general solution selecting the appropriate observation points and planning the path.

\subsection{Discretization of observation points} \label{discrete_points}

Provided that the objects are located in a continuous space, it is called for a method to generate discrete observation points that enables the UAV to practically plan its flying path in the finite space. It is inevitable that discretization brings a $1+\epsilon$ approximation error which is drawn from the method in~\cite{approximation_algorithm}. Specifically, this process will discretize the whole area into many mesh-points and the mesh granularity $\delta$ is calculated given Eq.~(\ref{mesh}) where $D = \mathop{\max}_{o_i,o_j \in O} d_{o_i, o_j}$ denotes the maximum distance between two objects.

\begin{equation}
    \delta = \frac{\epsilon * D}{4n} \label{mesh}
\end{equation}

\begin{lemma}\label{lemma:mesh}
The gross error of the discretization process does not exceed $\epsilon*D$.
\end{lemma}

\begin{proof}
Since mesh-points are used to approximate the points in the 2D continuous space, the approximation error for each point is at most $\delta$ when rounding it to the nearest mesh-point. The total error is less than $4n*\delta$ in that no more than $4n$ observation points are selected in the algorithms. The gross error does not exceed $\epsilon$ times of the lower bound of the optimal solution, so that the true relative error is at most $\epsilon$. $D$ is a satisfying lower bound of the problem since the UAV has to travel at least this distance to efficiently observe two objects and that we do not consider the situation with one object. Therefore, Lemma~\ref{lemma:mesh} is proved. By setting $4n*\delta = \epsilon*D$, we obtain Eq.~(\ref{mesh}).
\end{proof}

The obtained mesh granularity is then used to generate all the feasible observation points in the space. The process is described in Algorithm~\ref{algo:generate_observation_points}. For each object $o_i$, padding with width $d_{max}$ is applied to all four sides, followed by identifying all mesh points in the padded area. However, not all the mesh points are observation points because some of them are not efficient for the UAV, as defined in Definition~\ref{def:efficient_observe}. These infeasible points are then removed, leaving only the effective observation points. The representation of effective observation points $P$ are denoted by $ \{p_1,\, p_2,\, p_3,\, ... ,\,p_m\} $ and the coordinate of the $i^{th}$ point is $(x_{p_i},\,y_{p_i})$. The mesh result is illustrated in Figure~\ref{img:mesh_result}, in which the blue lines confine the possible area of the observation points and the effective ones are depicted in green, whereas the black points do not satisfy Definition~\ref{def:efficient_observe}.

\begin{algorithm}[h]
  \caption{Generation of observation points} % 名称
    \label{algo:generate_observation_points}
    \begin{algorithmic}[1]
      \Require
        Object set $O$,  observation range $d_{max}$, mesh granularity $\delta$
      \Ensure
        $P$: Observation points 
        \State $P \gets \emptyset$
        \For{object $o_i$ in $O$}
            \State Pad area with width $d_{max}$ on each of the four sides of $o_i$.
            \State Grid the padding area with granularity $\delta$ and gather the grid points as $P'$.
            \State Eliminate the points in $P'$ which do not satisfy Definition~\ref{def:efficient_observe}.
            \State $P \gets P \cup P'$
        \EndFor
        \State return $P$
    \end{algorithmic}
\end{algorithm}

\begin{figure}[htbp]
\centerline{\includegraphics[width=0.4\textwidth]{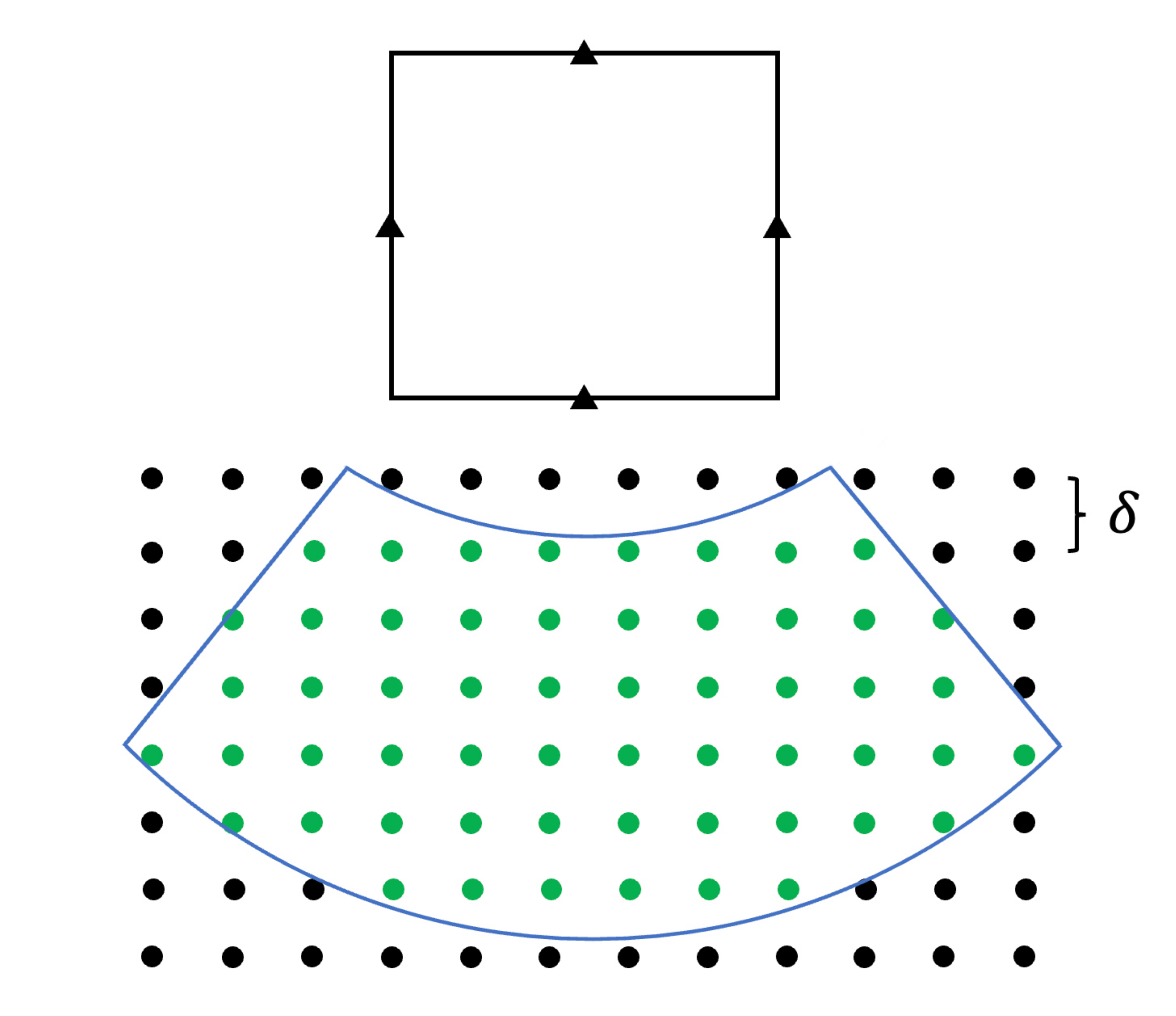}}
\caption{Area discretization of the observation range}
\label{img:mesh_result}
\vspace{0.3\in}
\end{figure}

\subsection{Integer Linear Programming representation}
\label{ILP_section}
In this section, we present the problem in an integer linear programming (ILP) form.
Given the discrete observation points in Section~\ref{discrete_points}, the ILP is devised based on the idea that the map is divided into several independent zones, each of which contains one side of an object and observation points that can cover it. As we have a starting point, it is regarded as an independent zone that contains only one point itself. The next step is to select one point in each zone and form a path passing all these observation points and a starting point. The designed ILP is based on the TSP formulation, and is shown in Eq.~(\ref{eq:0}) to Eq.~(\ref{eq:6}).

\begin{figure*}
\begin{align}
    \min \quad &\sum d_{i,j,p_1,p_2}\cdot X_{i,j,p_1,p_2}  \label{eq:0} \\ 
    s.t. &\sum_{j}\sum_{p_1}\sum_{p_2} X_{i,j,p_1,p_2} = 1, \forall i \in [1, 4n+1] \label{eq:1}\\
    &\sum_{i}\sum_{p_1}\sum_{p_2} X_{i,j,p_1,p_2} = 1, \forall j \in [1, 4n+1] \label{eq:2}\\
    &\sum_{j}\sum_{p_2} X_{i,j,p_1,p_2} = \sum_{j}\sum_{p_2} X_{j,i,p_2,p_1} \forall i \in [1, 4n+1], p_1\in zone~1 \label{eq:3}\\
    &\sum_{j \notin \{1, N+1\}}\sum_{p_2} (X_{1,j,p_1,p_2}+X_{N+1,j,p_1,p_2}) \notag \\
    & \qquad = \sum_{j \notin \{1, N+1\}}\sum_{p_2} (X_{j,1,p_2,p_1}+X_{j, N+1,p_2,p_1} ), \forall p_1\in zone~1  \label{eq:4}\\
&u_{i}-u_{j} + NX_{ijp_1p_2} \leq N-1,  \forall i, j \in V, i\neq j \neq 0 \label{eq:5}\\
 &X_{ijp_1p_2} \in 	\left\{ 0,1 \right\}, u_{i} \geq 0, u_{i} \in R \label{eq:6}
\end{align}
% \caption{ILP formulation}
\label{ILP_formula}
\end{figure*}

Let $X_{i,j,p_1,p_2}$ equal 1 if there is an edge between $p_1$ in zone $i$ and $p_2$ in zone $j$, hence the TSP path is composed of all the edges which are selected given $X_{i,j,p_1,p_2} = 1$. $d_{i,j,p_1,p_2}$ denotes the distance between point $p_1$ in zone $i$ and $p_2$ in zone $j$. Therefore, the objective function~(\ref{eq:0}) is to minimize the total cost of chosen edges with the following constraints. Eq.~(\ref{eq:1}) ensure that there is only one entered edge coming from other zones connecting only one point in zone $i$, and similarly, Eq.~(\ref{eq:2}) ensure one out edge leaves zone $i$ and goes to some other zone. To further restrict there is only one point chosen in each zone, the point connected with in-edge is supposed to be identical to that connected with out-edge, which is formulated in Eq.~(\ref{eq:3}) where the in-degree is always equal to the out-degree. Consequently, the point is either not chosen or has both in and out edge. As a supplement, Eq.~(\ref{eq:4}) aim to guarantee the correctness of the degree of the first and last point in the path. Besides, an effective TSP path requires the absence of sub-loop, which is eliminated by modifying the Miller-Tucker-Zemlin (MTZ) constraint~\cite{MTZ-tsp}, as presented in Eq.~(\ref{eq:5}).

\section{Methodology}
\label{methodology}
This section includes the complete methods solving both the offline and the online problems.

\subsection{An approximation algorithm} \label{offline_algrithm}

The idea of the proposed algorithm is as follows. For the offline solution, the algorithm works in two phases, the first of which seeks to identify the observation points that are to be visited and the second constructs a tour that passes all the selected observation points.  

Given a set of observation points, the offline algorithm begins by constructing a global graph, in which the information and relationships between observation points and objects are included. We then screen out the observation points as the way-points connecting all the objects' sides with relatively low cost. The proposed method achieves an approximation ratio of $(1+\epsilon)(2+2n)$, i.e., the obtained solution has a cost at most $(1+\epsilon)(2+2n)$ times of the cost of the optimal solution, as presented in Theorem~\ref{theorem1}.

\subsubsection{Construction of graph}
Recall that all the sides of the objects need to be observed and we use the midpoint of a side to represent it. The vertices of the graph is composed of the observation points $P$ and the midpoints of sides $Q$, with edges embodying the relations among the vertices. As shown in Algorithm \ref{graph}, if the UAV can efficiently observe a side $q_i$ of an object at location $p_j$, an edge with weight $D/2$ is connected between these two vertices. The motivation of setting the weight to be $D/2$ is to avoid the appearance of a path that links two observation point $p_m$, $p_n$ while passing one side, since the UAV cannot reach such a position. The second part is to add edges between each pair of observation points $p_i$ to $p_j$ with cost $d_{p_i,p_j}$ and edges between each observation points and the starting point with their distance. As a result, the expected graph $G$ is obtained.

\begin{algorithm}[htbp] 
  \caption{Graph construction} % 名称
  \label{graph}
  \begin{algorithmic}[1]
    \Require
       Observation points $P$, midpoints of each side of objects $Q$ and the starting point $p_{start}$.
    \Ensure
    Graph $G$
    \State Set as vertices $P$, $Q$ and $p_{start}$ 
    \For{side $q_i$ in $Q$}
        \For{observation point $p_j$ in $P$}
            \If{$p_j$ can efficiently observe side $q_i$}
                \State add an edge with weight $ D/2 $
            \EndIf
        \EndFor
    \EndFor

    \For{each pair of points $p_i$ and $p_j$}
        \State add edge with weight $d_{p_i,p_j}$
    \EndFor
    \For{each point $p_i$ in $P$}
        \State add edge between $p_i$ and $p_{start}$ with weight $d_{p_i,p_{start}}$
    \EndFor
    
    \State return  $G$.
  \end{algorithmic}
\end{algorithm}

\subsubsection{Selection of the observation points}
The next step is to connect the vertices of the graph via a structure that both ensures connectivity and the lowest cost. Granted that directly finding the minimum spanning tree on graph $G$ will connect all the objects and observation points, it does not take into account the internal relationship that the UAV may be able to observe several sides of objects at one location, hence there may be unnecessary to visit all of the observation points. In contrast, it is desired to search for a tree connecting all the sides of objects and some observation points that exactly cover them, leading to the essential idea of our algorithm: reducing to the Steiner Tree problem.

\begin{definition} [Steiner Tree] \label{def:Steiner-Tree} 
Given an undirected graph ${G} = (V, E)$ with non-negative weights, a Steiner tree is a tree that spans the Steiner points $\mathcal{S}$, which is given as input where $\mathcal{S} \subseteq V$. 
\end{definition} 

The minimum Steiner Tree is a Steiner Tree with the minimum total of edges' weights, which is known to be a NP-complete problem~\cite{Steiner-tree-np-hard}. Accordingly, a way to solve for a Steiner Tree is through approximation algorithms, one of which is described below and achieves an approximation ratio of 2. In our situation, we seek for a Steiner Tree of the global graph $G$ with this algorithm by setting all the starting point and the sides of objects $Q$ as Steiner points $S$.

For completeness, we describe the algorithm~\cite{Steiner-tree} for finding a Steiner Tree. It first computes the shortest paths between each pair of Steiner points, the aggregation of which yields a graph $G_1$. Consequently, $G_1$ is a sub-graph of $G$, containing only all the Steiner points $Q$ by reducing some unrelated edges. The next step is solving the minimum spanning tree $T_s$ of $G_1$, which further narrows the graph's size. At this moment, $T_s$ is actually no larger than the minimum spanning tree of the initial graph $G$ but probably still includes some unnecessary edges and vertices. Recall that the desired Steiner tree only needs to span Steiner points, hence the last step is to delete redundant edges and vertices, ensuring all the leaves of the tree are Steiner points.

A possible Steiner Tree is presented in green lines (both solid and dashed) in Figure~\ref{img:possible_steiner_tree} that contains three objects, each of which is surrounded by according observation points in the same color. The starting point and Steiner points are depicted by the large orange triangle and the small triangles on each sides. The green dashed lines show the relationship between the Steiner points and the selected observation points.

\begin{figure}[htbp]
\centerline{\includegraphics[width=0.8\textwidth]{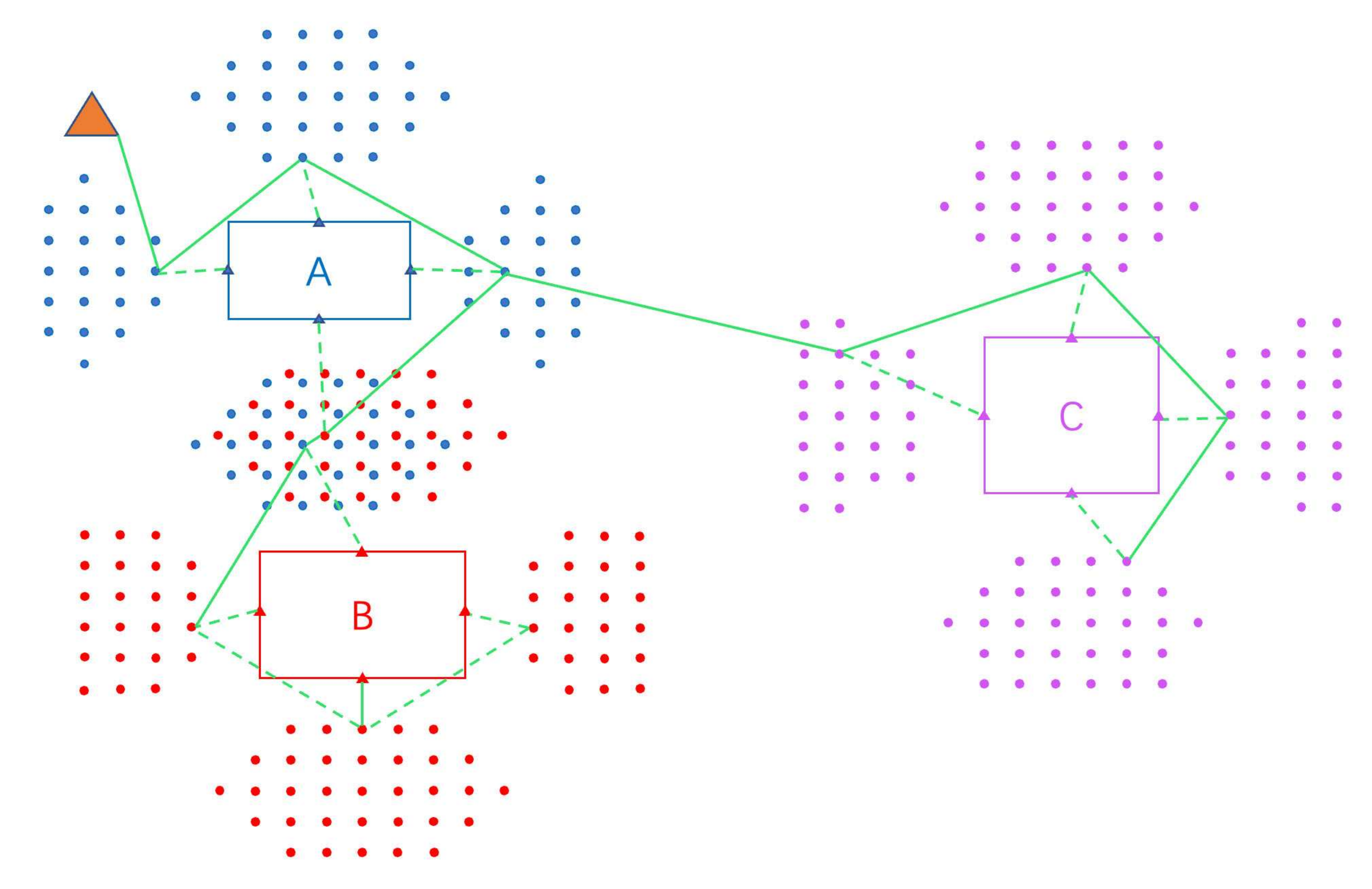}}
\caption{A possible Steiner tree connecting the selected observation points that observes all faces of all objects}
\label{img:possible_steiner_tree}
\end{figure}

\subsubsection{Construction of the path} \label{sec:form path part}
After obtaining the Steiner tree, we still need to remove the edges between the Steiner points and the selected observation points since the UAV does not need to go through these edges to observe the objects. The tree after trimming is the one that can be correlated with the UAV's path. Up to now, the whole algorithm has almost been completed, leaving only the formation of the path for UAV. Therefore, based on the trimmed Steiner tree, a set of observation points is collected, including the starting point. The acquired Steiner tree chooses key observation points with low cost on edge while satisfying the coverage constraints. Denote these observation points as $R$ where $R \subseteq P $ and all of them need to be visited by the UAV. We carried out the algorithm for TSP to find the shortest path on these way-points. Note that TSP is proven to be an NP-hard problem, and we use the 1.5-approximation algorithm that refers to~\cite{1.5TSP}.

\begin{algorithm}[h]
  \caption{Approximation algorithm} % 名称
    \label{algo:our_appro_algo}
    \begin{algorithmic}[1]
      \Require
        Object set $O$
      \Ensure
        $L$: the length of TSP path
      \State Obtain observation points with Algorithm~\ref{algo:generate_observation_points}.
      \State Construct a graph with Algorithm~\ref{graph}.
      \State Construct a Steiner Tree based on~\cite{Steiner-tree}.
      \State Calculate the tour cost based on~\cite{1.5TSP}.
    \end{algorithmic}
\end{algorithm}

\begin{lemma}
    \label{blue>D}
    The cost of the Steiner tree after trimming $T_{left}$ connecting the observation points is at least $D$.
\end{lemma}

\begin{proof}
    %The proof of Lemma \ref{blue>D} is simple. 
    Considering that the Steiner tree connects all the objects by choosing only one observation point from each of them, there must exist a path in the tree that connects two farthest objects. Since the path is not shorter than straight line between these two points, thus the cost of the Steiner tree is at least greater than $D$.
\end{proof}

We use $T_{trimmed}$ to denote the part of the edges that were removed and $LB$ the lower bound cost.

\begin{lemma}\label{lemma2}
    $LB(T_{trimmed}+T_{left})-2nD \leq LB(T_{left})$
\end{lemma}

\begin{proof}
As described above, the weight of the edge between an observation point and a side of an object is $D/2$. Therefore, the weight of $T_{trimmed}$ equals to $2nD$. By unfolding the inequality, the lemma follows.
\end{proof}

\begin{lemma}\label{lemma3}
    $LB(T_{trimmed}+T_{left})\leq  T_{left}+2nD \leq 2*LB(T_{trimmed}+T_{left})$
\end{lemma}

\begin{proof}
By lemma~\ref{lemma2}, we have $LB(T_{trimmed}+T_{left})\leq  T_{left}+2nD$. Since the algorithm for solving the minimum Steiner Tree problem is 2-approximate~\cite{Steiner-tree} and the cost of the Steiner is $T_{left}+2nD$, we have $T_{left}+2nD \leq 2*LB(T_{trimmed}+T_{left})$. Therefore, the lemma follows.
% This lemma demonstrates the results of the Steiner tree algorithm. The lower bound is exactly $LB(T_{trimmed}+T_{left})$ and the cost of the Steiner tree is actually equal to $T_{left}+2nD$. Since Steiner tree algorithm is 2-approximate, the cost of the Steiner tree is less than $2*LB(T_{trimmed}+T_{left})$. Therefore, Lemma \ref{lemma3} is proved.
\end{proof}

\begin{theorem}
\label{theorem1}
    Algorithm \ref{algo:our_appro_algo} has an approximation ratio of $(1+\epsilon)(2+2n)$.
\end{theorem}
\begin{proof}
    The overall cost of the Steiner tree is not greater than $2*LB(T_{left})+2nD$ and the lower bound of the problem is $LB(T_{left})$. Therefore:

    \begin{align}
        \frac{ALG}{OPT} = (1+\epsilon)*\frac{2*LB(T_{left})+2nD}{LB(T_{left})}  &= (1+\epsilon)*\left(2+\frac{2nD}{LB(T_{left})}\right)\notag \\
        &\leq (1+\epsilon)*(2+\frac{2nD}{D})\notag \\
        &= (1+\epsilon)*(2+2n)\notag 
    \end{align}
\end{proof}

\subsection{Online solutions}
In the online setting, we suppose the UAV is initiated at some location with several objects around it. There are other unknown objects in the map as a result of the limitation of UAV's perception range. The problem is then transformed into covering all the objects with limited information while exploring the unknown area. The difference from the offline algorithm is that the UAV needs to continuously update its known information and adjust the flying path when detecting new objects. We designed three algorithms to address the problem. The first step is meshing the known map with Eq.(~\ref{mesh}). $D$ here is the maximum distance between any objects in the initially known map. This procedure also ensures an approximation ratio of $1+\epsilon$, since $D$ in the initially known map is not greater than that in the globally known map. Therefore, we get a smaller mesh size in the online mode and the approximation ratio is then guaranteed. Furthermore, it is assumed that there is no object locating very far from the others so that it cannot be detected by the UAV during its flight. 

\subsubsection{Nearest Object First (NOF)}
The first online algorithm takes the intuitive idea that the UAV always chooses the nearest object and observes its sides (Algorithm~\ref{algo:nearest}). As included in the outer loop, the UAV selects the nearest side of an object $q_i$ and visits it by flying to the closest observation point $p_j$. Before $p_j$, it continuously updates its known information by appending newly-detected object (including objects' sides) and observation points into set $Q_{known}$ and $P_{known}$. Then $q_i$ and other sides $q_r$ are deleted from $Q_{known}$ that can be covered by UAV at $p_j$. The process continues until $Q_{known}$ is empty, meaning that all the objects have been covered. During the process, the location of UAV is recorded ceaselessly, from which length of the entire path is computed.

\subsubsection{Cheapest Insertion (CI)}
The second online algorithm draws the idea from the cheapest insertion algorithm~\cite{cheapest-insertion} which computes a tour by iteratively choosing an observation point and inserting it into the path that minimizes the insertion cost. Algorithm~\ref{algo:cheapest} describes the procedure. Different from the first online algorithm, this method begins by finding a path from the initial known area with the front part of offline method. Solve the Steiner tree of the known map and conduct TSP algorithm on the tree. This closed path ensures to cover all the objects and their sides in the primary map. Then the UAV starts to travel along the path as that in Algorithm \ref{algo:nearest}. The method makes a distinction in the way of adjusting path by inserting the observation points into the current path at the lowest cost. Since the UAV may efficiently observe some sides at newly discovered location $p_i$, there is no need to consider these sides when inserting their corresponding observation points. Similarly, the objects' sides are deleted when reaching the next point $p_i$ which is then popped from $Path$. The algorithm ends when $Path$ is empty and the path length is computed, meaning that all the demanding observation points have been visited.

\begin{algorithm}[htbp]
  \caption{Nearest object first} % 名称
  \label{algo:nearest}
  \begin{algorithmic}[1]
    \Require
       currently known objects' sides $Q_{known}$,
       currently known observation points $P_{known}$
    \Ensure
       $L$: the length of TSP path
    \While{$Q_{known}$ not empty}
        \State select the nearest $q_i$ from $Q_{known}$ 
        \State select the nearest $p_j$ from $P_{known}$ that can observe $q_i$
        \While{not reach $p_j$}
            \State move for a step and record current position
            \State update known information
            \If {detect a new nearest $q_z$}
                \State $Q_{known} \gets Q_{known} \cup q_z$
            \EndIf
        \EndWhile
        \State delete $q_i$ from $Q_{known}$
        \If{UAV covers other $q_r$}
        \State delete other $q_r$
        \EndIf
    \EndWhile
    \State Back to the starting point
    \State \Return The path length ${L}$ given the record
%    \State return ${L}$
  \end{algorithmic}
\end{algorithm}

\begin{algorithm}[htbp]
  \caption{Cheapest insertion} % 名称
  \label{algo:cheapest}
  \begin{algorithmic}[1]
    \Require
       currently known objects' sides $Q_{known}$,
       currently known observation points $P_{known}$
    \Ensure
       $L$: the length of path
ind a tour $Path$ for the initial known map 
        \While{$length(Path)\geq 2$}
        \State $p_{next} \gets Path[1]$
        \While{not reach $p_{next}$}
        \State move for a step and record current position
        \State update known information
        \If{detect new objects and sides $Q_{new}$}
        \For{each $q_j$ in $Q_{new}$}
        \State insert observation point $p_a$ that observes $q_j$ into $Path$ with lowest cost
        \State delete $q_m$ from $Q_{new}$ that covered by $p_a$
        \EndFor
        \EndIf
        \EndWhile
        \State delete $q_j$ from $Q_{known}$ covered at $p_{next}$
        \State delete $Path[0]$
    \EndWhile
    \State \Return The path length ${L}$ given the record
    %\State return ${L}$
  \end{algorithmic}
\end{algorithm}

\subsubsection{Best Available TSP}
The last online method (Algorithm~\ref{algo:update_tsp}) is based on the 1.5-approximation TSP algorithm and adjusts the path by updating it when detecting new objects. The overall process bears very close resemblance to cheapest insertion, with the only difference in updating the path. Once UAV has discovered new objects and observation points in the new area, the Steiner tree method is used again to choose some observation points covering the new area, the aggregation of which with those in the current path is used to solve a new TSP path. We use array $Path$ to dynamically store the sequence of observation points that the UAV is to visit and $Path[i]$ is the $i^{th}$ point of the current $Path$. Considering the order of the new path may be different, there is need to change the new path by adjusting $Path[0]$ as the starting point. At this moment, the next point could be different so that the it should be updated. The algorithm ends when $Path$ is empty or the initial point is the only observation point left, the latter implying that the UAV can directly return back.

\begin{algorithm}[h]
  \caption{Best available TSP} % 名称
  \label{algo:update_tsp}
  \begin{algorithmic}[1]
    \Require
       currently known objects' sides $Q_{known}$,
       currently known observation points $P_{known}$
    \Ensure
       $L$: the length of path
       \State Compute a Steiner tree $T$ the initial known map 
       \State Compute a closed TSP path $Path$ with respect to $T$
        \While{$length(Path)\geq 2$}
        \State $p_{next} \gets Path[1]$
        \While{not reach $p_{next}$}
        \State move for a step and record current position
        \State update known information
        \If{detect new objects' sides $Q_{new}$ and new observation points $P_{new}$}
            \State $G_{local} \gets $ local graph  on $Q_{new}$ and $P_{new}$
            \State $T_{local} \gets $ Steiner tree  on $G_{local}$
            \State $P_{total} \gets $ aggregate observation points in $T_{local}$ and $Path$
            \State $Path \gets $ TSP on $P_{total}$ with start point of $Path$
            %\State  Path_{new}$
        \EndIf
        \State $p_{next} \gets Path[1]$
        \EndWhile
        \State delete $Path[0]$
       
    \EndWhile
    \State \Return The path length ${L}$ given the record
    %\State return ${L}$
  \end{algorithmic}
\end{algorithm}

\section{Experiments}
\label{experiments}
In this section, we evaluate the performance of the proposed algorithms via carrying out a series of simulated experiments.

\vspace{0.2\in}
\subsection{Experiments setting}

A number of objects are placed randomly on a map with side length 120 meters (m). We applied padding with width 10m to each side of the map in order to keep the coordinates positive of all the objects and observation points. The objects' sizes are chosen from 1m*2m, 2m*2m, and 1m*1m. To ensure both the quality and safety of observation, we set the maximum and minimum observing distances to be 4m and 1m, respectively, with the visual angle of the UAV being 120 degrees. In the online setting, the perception range of the UAV is set to 40m. To avoid the situation that some objects are omitted because they cannot be detected by the UAV, we assume that for each object $o_i$ there exists at least another object $o_j$ whose distance is less than the perception range.

In the experiments, we set the number of objects $n$ to $\left\{5, 10, 15, 20, 25\right\}$ and the $\epsilon$ to $\left\{0.1, 0.2, 0.3, 0.4, 0.5\right\}$. The time limit is set to 30000s and $\epsilon$ to 0.2 for Gurobi~\cite{Gurobi} (an ILP solver). After fixing the number of objects $n$, we run 20 test cases by randomly generating the corresponding objects and testing all algorithms for each epsilon. Since the running time is unacceptable when $n$ is 25 and $\epsilon$ is 0.1, we do not test this pair of parameters.

All the experiments are conducted on the Linux server, with Intel (R) Xeon (R) Gold 5215 CPU 2.50GHz and 188 GB RAM.

\subsection{Integer linear programming and lower bound}
We use Gurobi to solve the designed ILP presented in Section~\ref{ILP_section}, which is a large scale mathematical planning optimizer developed by Gurobi Company~\cite{Gurobi}. During the computation, it maintains a lower bound of the ILP as well as the best feasible solution. The latter is compared with the performance of the proposed algorithms in Section~\ref{sec:numerical_results}. Both values are then used for comparison in the simulations. 

In order to strengthen the approximation ratio, we use another theoretical lower bound, which is the Steiner tree cost of the graph after trimming its edges between the Steiner points and the selected observation points. It is worth noting that the Steiner tree does not have a precise solution so that the cost is computed through the 2-approximation algorithm from~\cite{Steiner-tree} and the lower bound is half of the result. Since Gurobi optimizer also produces a lower bound, these two are then compared and the higher one is selected to evaluate the approximation ratio.  

\subsection{Numerical results}
\label{sec:numerical_results}

In this part, the statistical results are presented in the following figures and tables. We first present the optimizing results of Gurobi in Section~\ref{sec:gurobi_appro_result}, followed by the approximation ratio of our algorithms in two parts. Specifically, Section~\ref{sec:path_cost1} analyzes the approximation ratio relative to the lower bound presented in the previous section, and Section~\ref{sec:path_cost2} focuses on the relative ratio compared to the best feasible results of Gurobi. In addition, the running time of our algorithms are presented with different values of epsilon in Section~\ref{sec:runningtime}. Finally, we show that the approximation ratio is relatively stable with different number of objects in Section~\ref{sec:stable appro ratio}.

\subsubsection{Optimization results of Gurobi} \label{sec:gurobi_appro_result}

\begin{figure}[ht]
\centering
\includegraphics[width=0.48\textwidth]{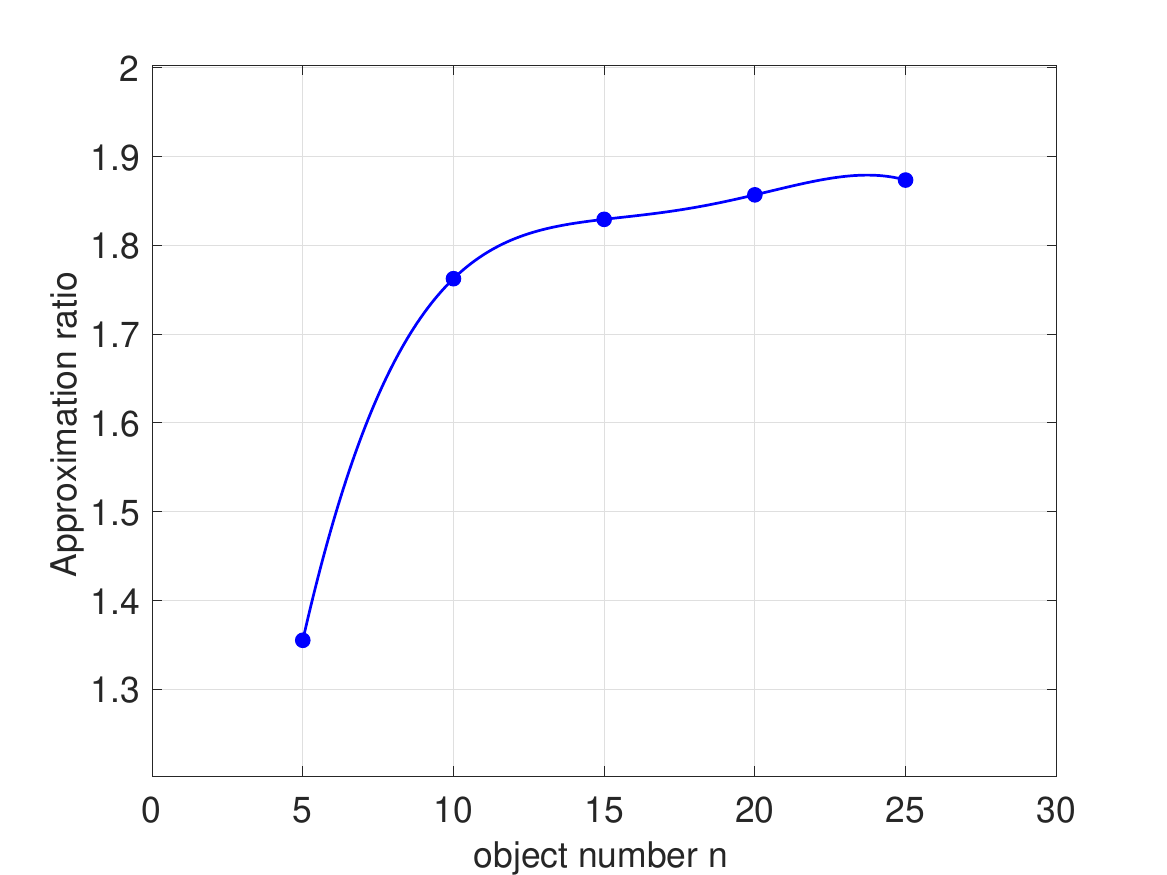}
\caption{Approximation ratio with Gurobi with $\epsilon=0.2$.}
\label{fig:gurobi_appro_result}
\vspace{0.5\in}
\end{figure}

 The approximation ratio of the solutions returned by Gurobi optimization is computed by the best feasible result of Gurobi over its lower bound. In other words, the smaller the gap is, the better the optimizing effect is. Figure~\ref{fig:gurobi_appro_result} shows that the approximation ratio is around 1.85 and is relatively stable with the increase of the object number from 10 to 25. There is an exception that the ratio is 1.35 when $n$ equals 5 which indicates that the Gurobi can optimize the objectives more effectively. For the rise in ratio, we found that the two optimization objectives (lower bound and best feasible solution) nearly remained constant in the second half of the optimization period when $n$ is greater than 10, owing to the fact that there are millions of parameters, which could be beyond Gurobi's computation capability. Considering the apparent gap between the best feasible solution and the lower bound of Gurobi, there is a need to evaluate our algorithms on both the two indices. 

\subsubsection{Approximation ratio compared to the lower bound} \label{sec:path_cost1}
The path costs of our algorithms are divided by the lower bound and the ratios are drawn in Figure~\ref{fig:tolowerbound}. The approximation ratio of NOF is around 2.2 when $n =$ 10, 15, 20 and 25. Specifically, NOF achieves the best performance when $n=15$, compared to the other three cases. The performance of CI algorithm decreases with the increase of $n$. Offline methods possess a relatively stable performance with the best shown in $n=5$.

It is noted that in all settings, Best Available TSP (BATSP) algorithm achieves the highest approximation ratio and is significantly higher than the other three algorithms. When there are 5 objects, the ratio fluctuates from 1.67 to 1.81 when $\epsilon$ ranges from 0.1 to 0.5. The ratio remains the similar trend when $n = 10$, beginning with a small increase to 2.89, then decreasing to the lowest value 2.71 and rebounding in the end. As for $n = 15$, the ratio rises from 2.74 to 3.08 before a small drop to 2.91 when $\epsilon$ is 0.5. In $n=20$ setting, the ratio continuously climbs from 2.78 to 3.20 and in $n=25$ setting, the ratio fluctuates again from 3.02 to 3.22.

\begin{figure}[h]
\centering
\begin{minipage}[b]{1\linewidth}
	\subfloat[object number = 5]{\label{fig:object5}
	\includegraphics[width=0.45\linewidth]{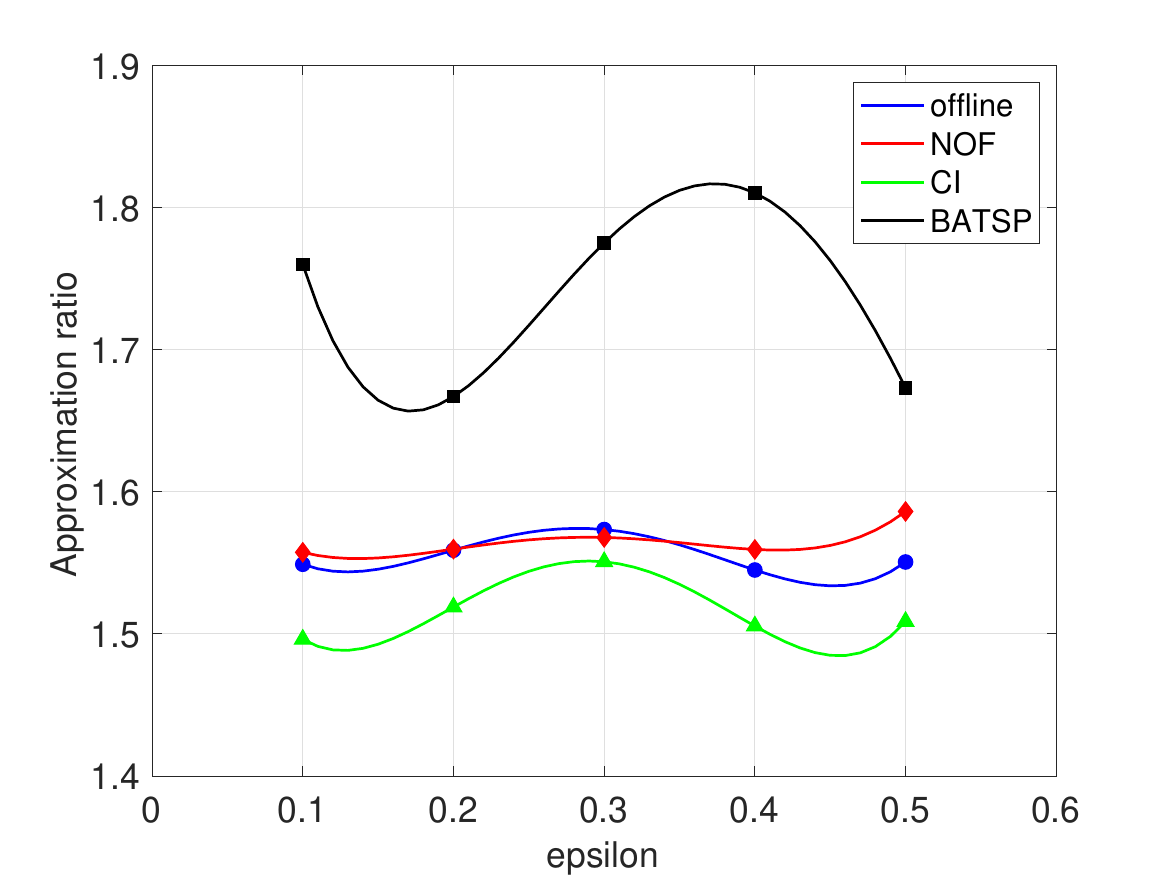}}
	\subfloat[object number = 10]{\label{fig:object10}
	\includegraphics[width=0.45\linewidth]{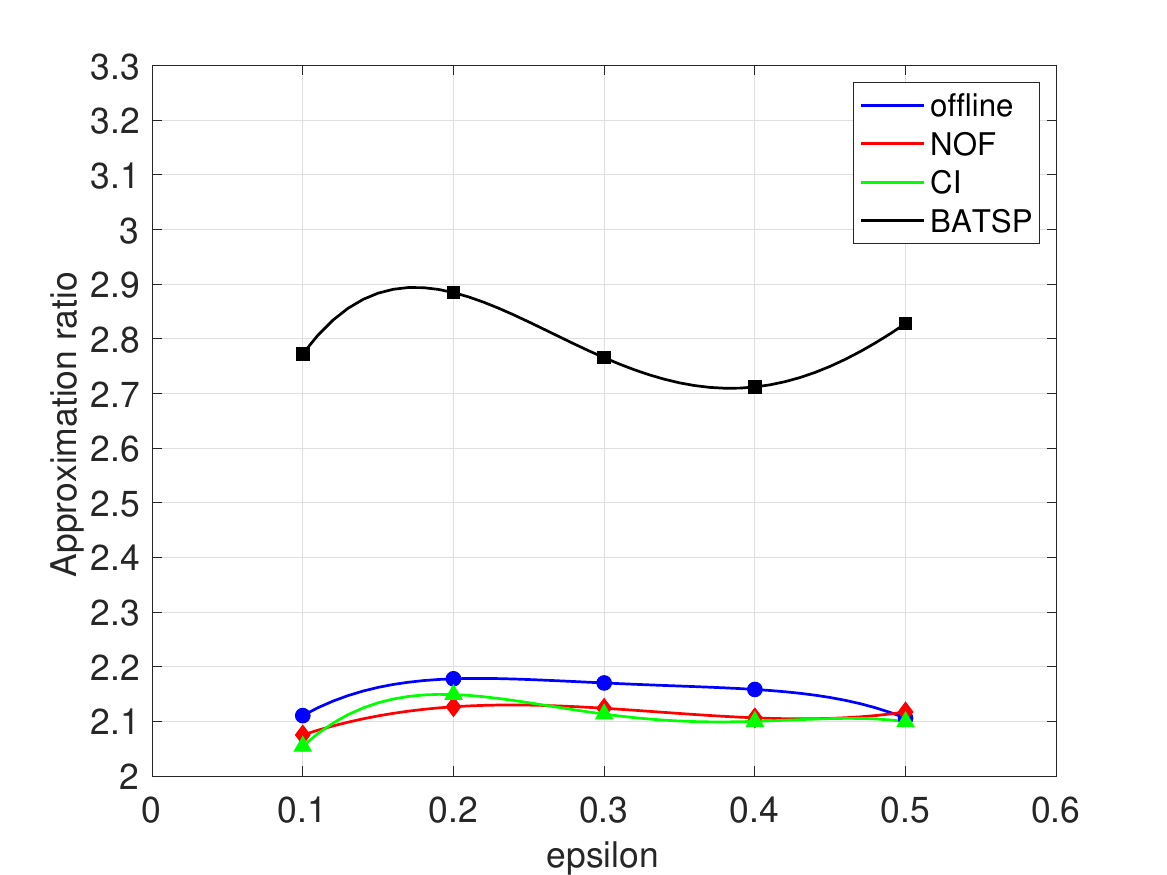}}
\end{minipage}
\begin{minipage}[b]{1\linewidth}
	\subfloat[object number = 15]{\label{fig:object15}
	\includegraphics[width=0.45\linewidth]{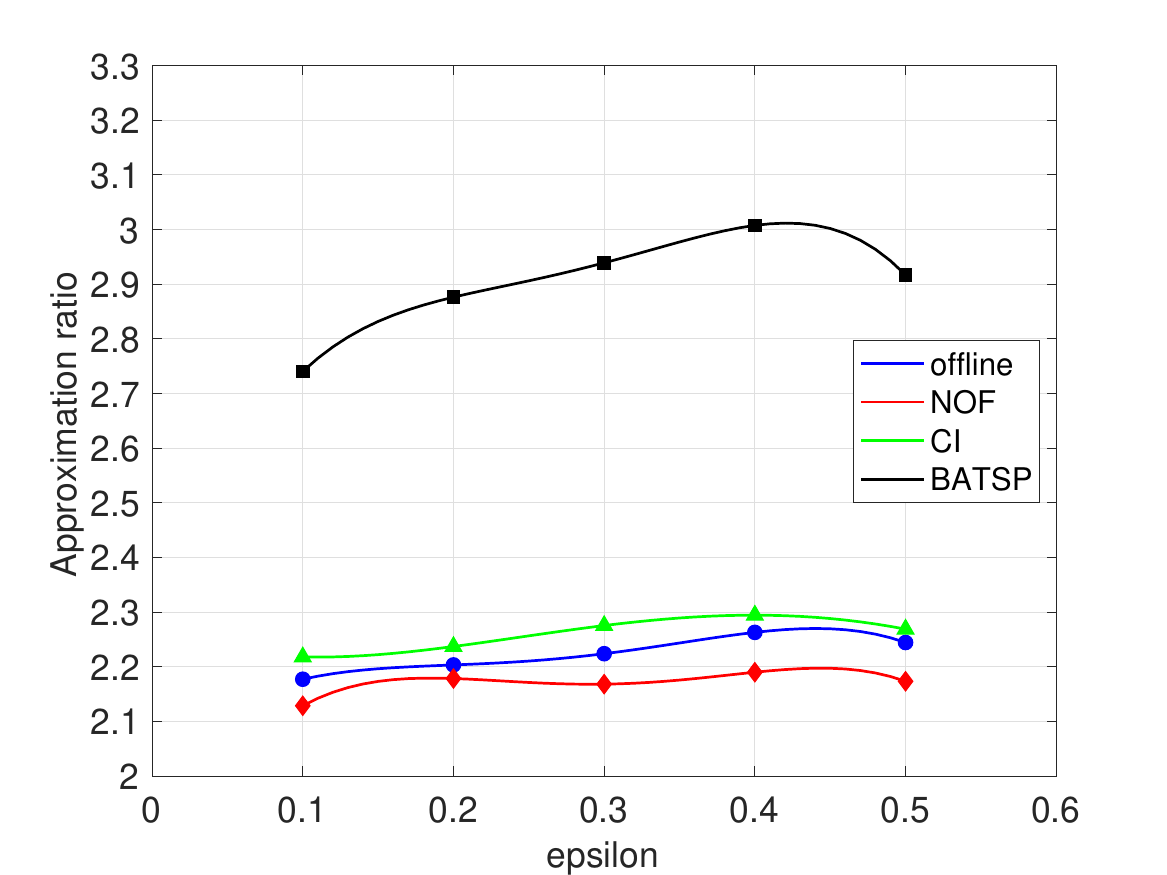}}
	\subfloat[object number = 20]{\label{fig:object20}
	\includegraphics[width=0.45\linewidth]{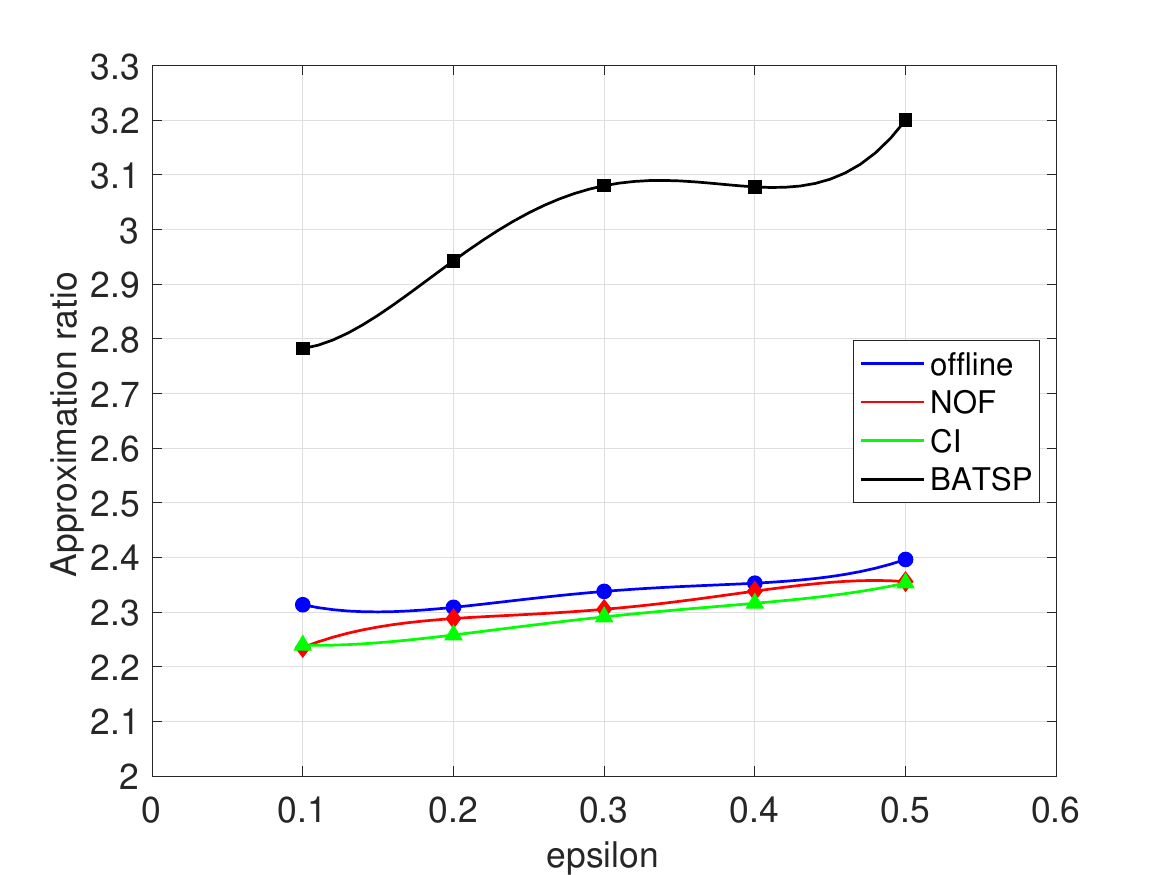}}
\end{minipage}
\begin{minipage}[b]{1\linewidth}
	\subfloat[object number = 25]{\label{fig:object25}
	\includegraphics[width=0.45\linewidth]{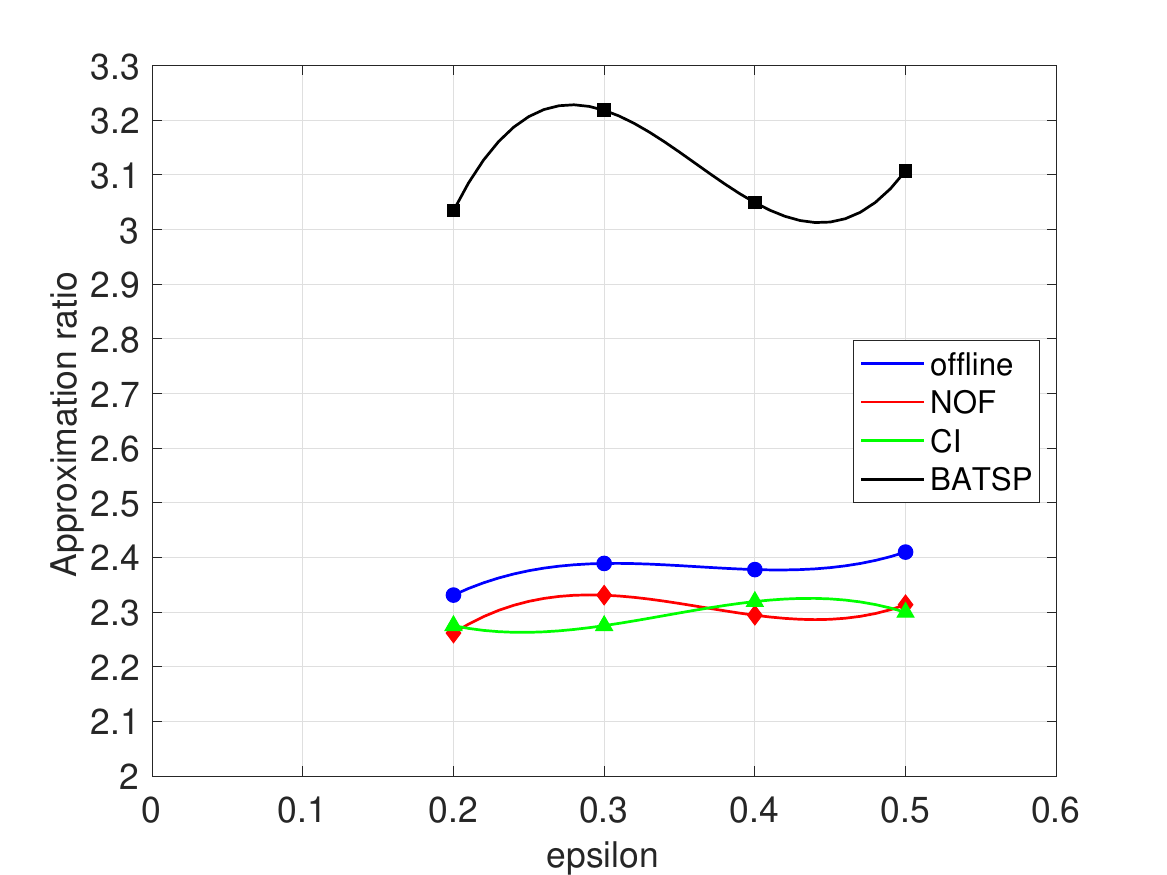}}
\end{minipage}
\caption{Approximation ratio related to obtained lower bound}
\label{fig:tolowerbound}
\end{figure}

\subsubsection{Approximation ratio compared to the best feasible solution} \label{sec:path_cost2}

Figure~\ref{fig:tobestresult} illustrates the ratios by dividing the cost of the path by the best feasible result of Gurobi, which are much lower than compared to the lower bound. Concretely, CI algorithm performs best when $n$ equals 5, 20 and 25, but the ratio surpasses that of the offline and NOF methods for $n=15$. Both NOF and offline methods are stable in all of the settings. In general, these three methods possess the ratio nearly 1.25, meaning that they are close to the best result computed by Gurobi. Moreover, BATSP method performs poorly compared to other methods, ranging from 1.22 to 1.35 when $n=5$, 1.5 to 1.64 when $n=10$ and $15$. For more complex situation, its ratio rises to around 1.73.

\begin{figure}[htbp]

\begin{minipage}[b]{1\linewidth}
	\subfloat[object number = 5]{\label{fig:gurobi5}
	\includegraphics[width=0.45\linewidth]{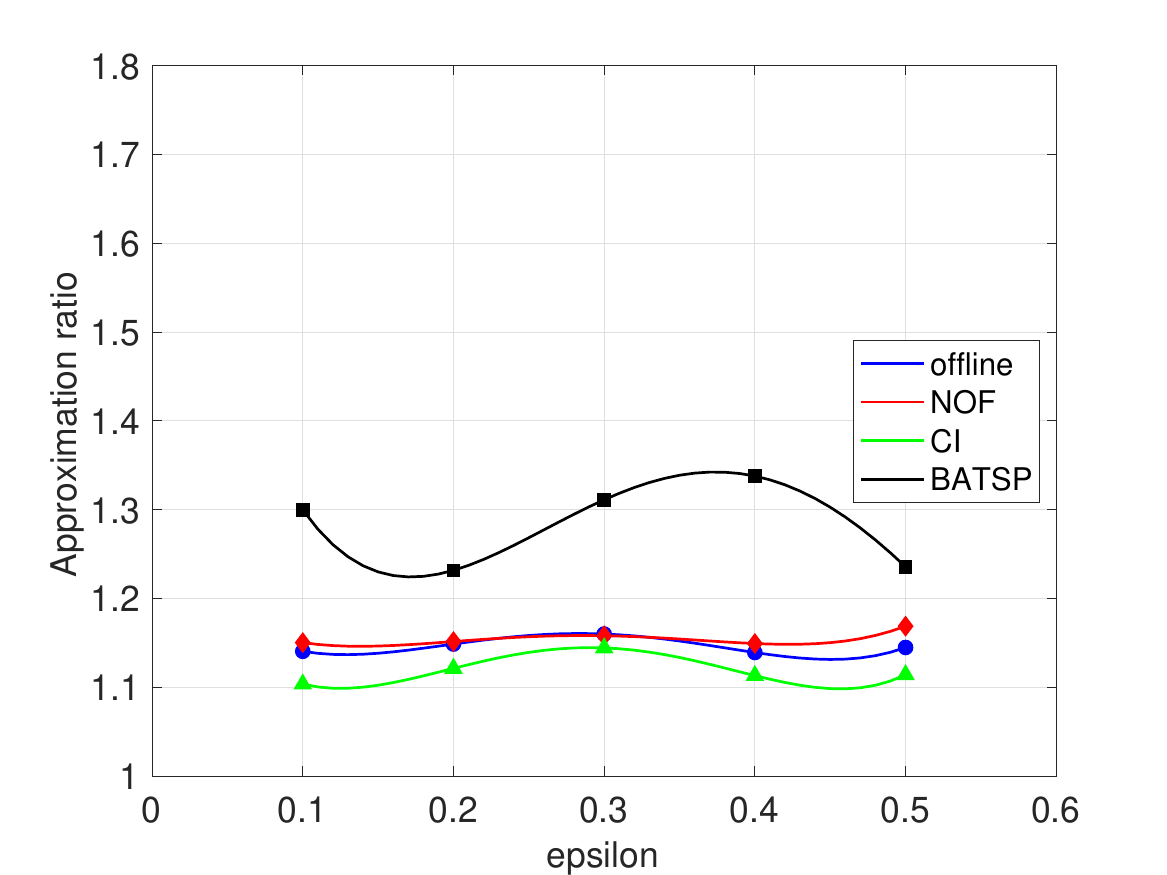}}
	\subfloat[object number = 10]{\label{fig:gurobi10}
	\includegraphics[width=0.45\linewidth]{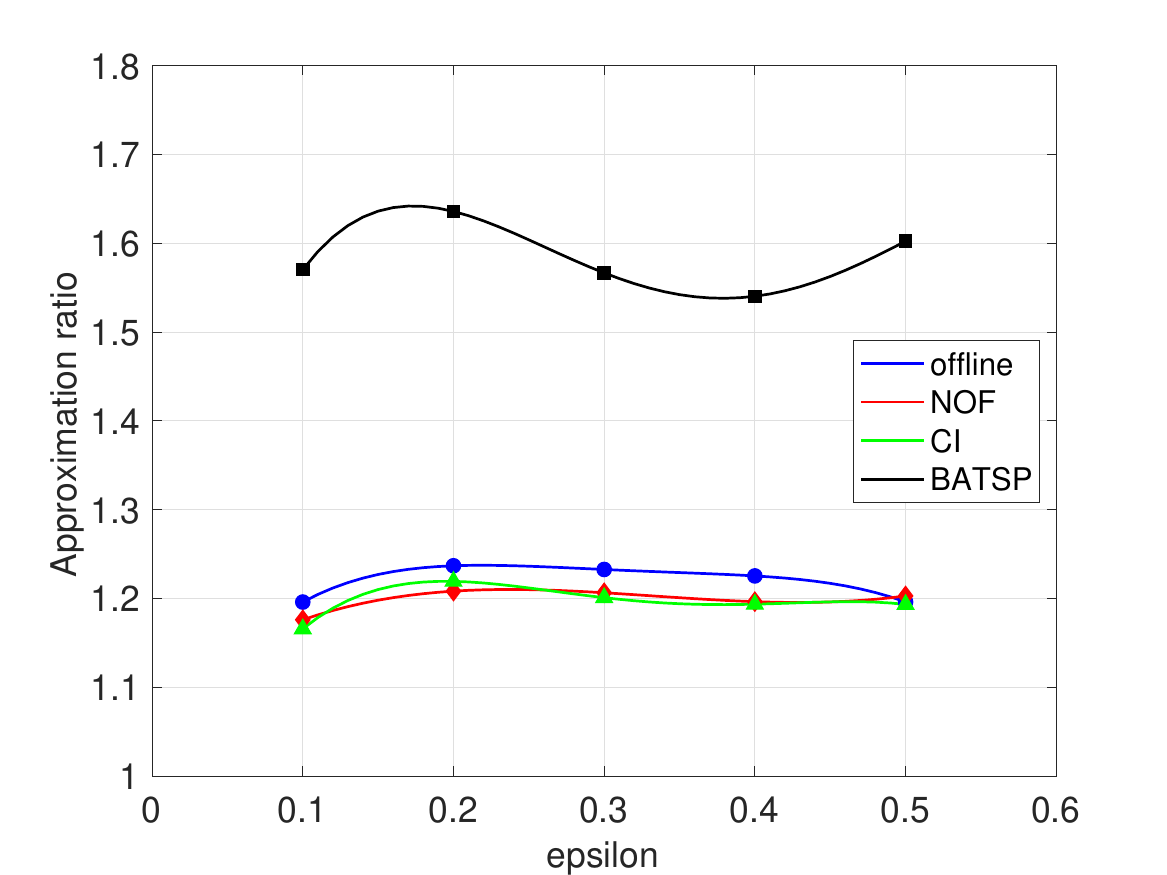}}
\end{minipage}
\\
\begin{minipage}[b]{1\linewidth}
	\subfloat[object number = 15]{\label{fig:gurobi15}
	\includegraphics[width=0.45\linewidth]{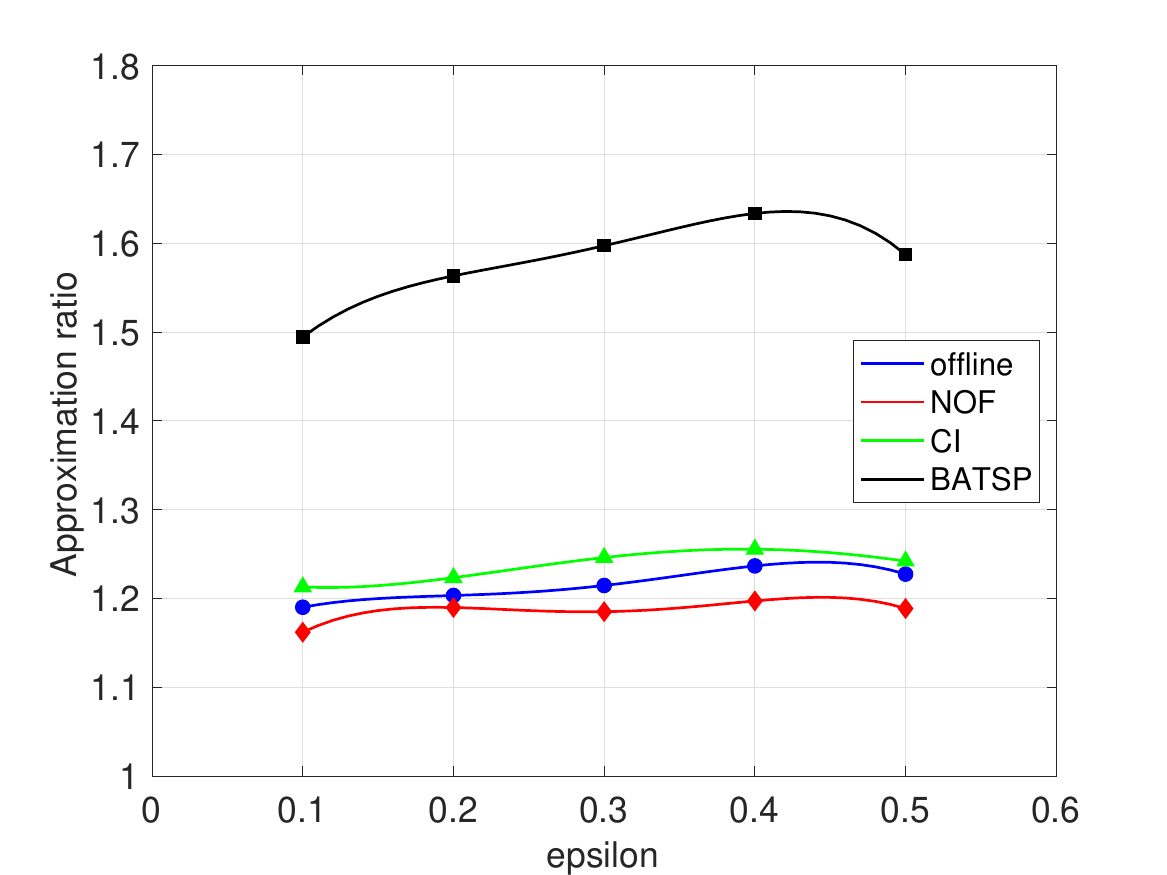}}
	\subfloat[object number = 20]{\label{fig:gurobi20}
	\includegraphics[width=0.45\linewidth]{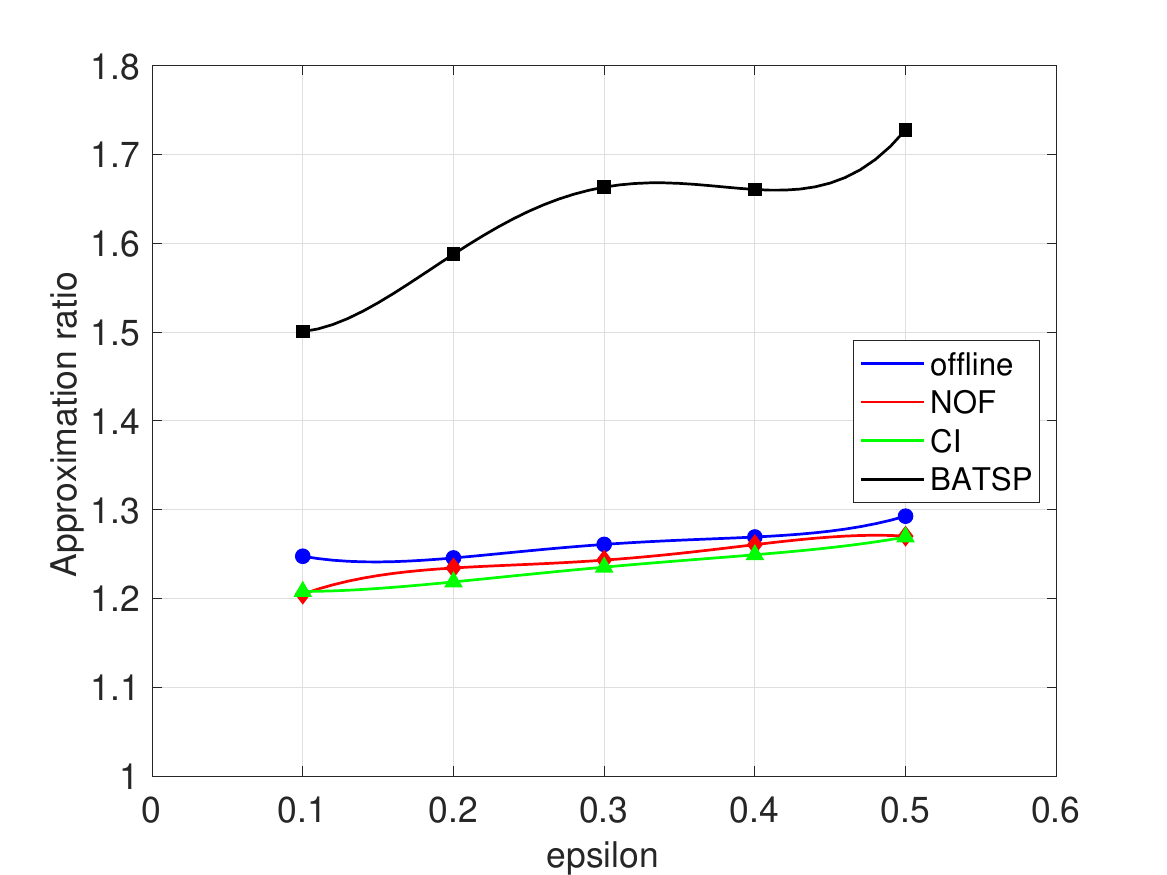}}
\end{minipage}
\\
\begin{minipage}[b]{1\linewidth}
	\subfloat[object number = 25]{\label{fig:gurobi25}
	\includegraphics[width=0.45\linewidth]{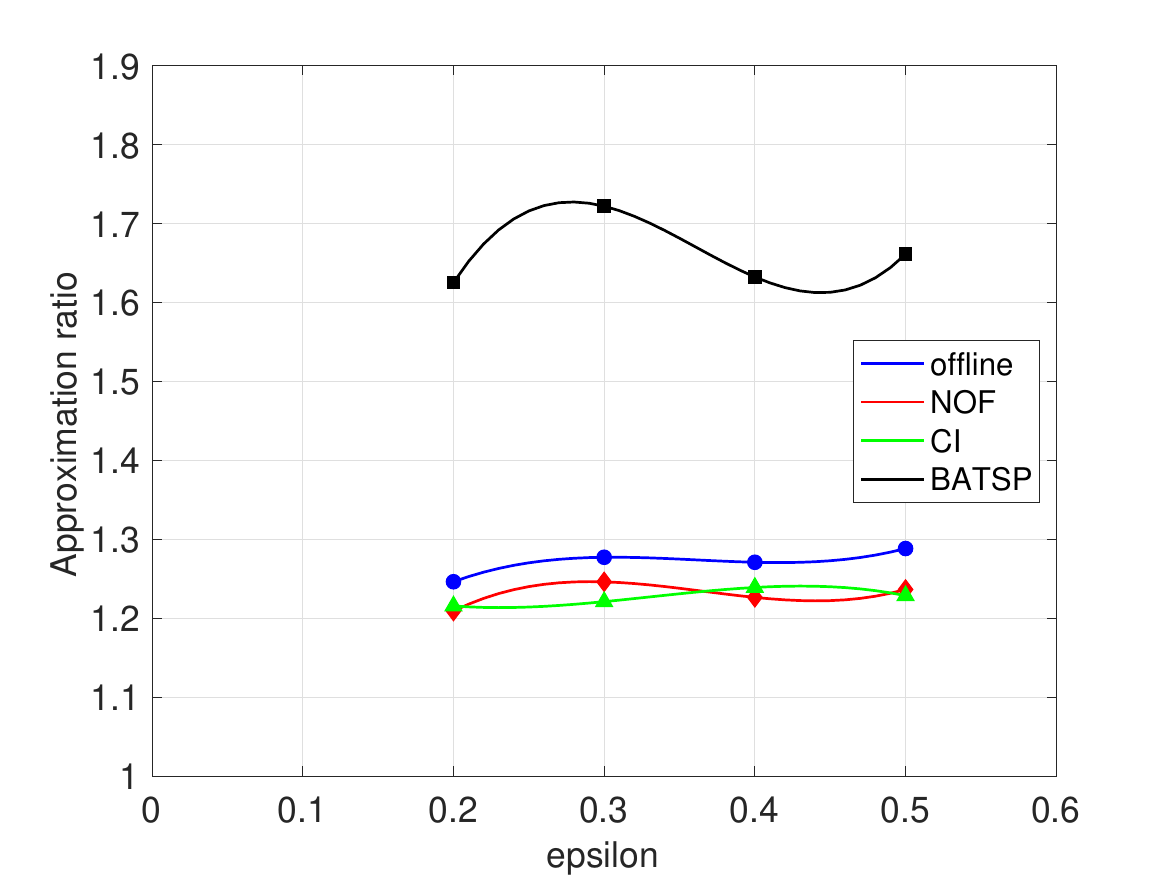}}
\end{minipage}

\caption{Approximation ratio compared to Gurobi best result }
\label{fig:tobestresult}
\end{figure}

\subsubsection{Running time} \label{sec:runningtime}
In addition to the path cost, the average running times are presented in Table~\ref{table:5objects} and \ref{table:epsilon02s} for all algorithm in different epsilon settings. All of the data units are seconds. Generally, the time cost surges exponentially with the increase of the number of objects or the decrease of epsilon, especially from 0.2 to 0.1. Among the four algorithms, the offline method takes the longest time to produce the results and NOF method in contrast boasts fastest speed. There is not much difference in the running time of CI and BATSP, with CI running slightly faster.

Concretely, for 5 number of objects, the average time of NOF is less than 0.2 second for all values of epsilon which is extremely efficient. The offline method takes 1.48 seconds with epsilon 0.3 whereas the remaining algorithms finish within 1 second with greater epsilon. However, the decrease of epsilon from 0.2 to 0.1 considerably rises the time of offline method to around 2240 seconds and CI and BATSP to around 987 seconds, around 200 times longer than that with epsilon 0.2. The change trend is similar for more objects. The offline method time with epsilon 0.2 rises to 302.63 seconds for 10 objects and further to more than a thousand seconds for 15, 20 and 25 objects. All algorithms have an acceptable running time when the number of objects is 25 and $\epsilon$ is greater than 0.2, but the NOF algorithm can achieve very close performance to other algorithms in an extremely short time.

\begin{table}[]
\centering
\caption{Running time for 5 objects (in seconds)}\label{table:5objects}
\begin{tabular}{crrrr}
\hline
$\epsilon$ & offline    & NOF        & CI          & BATSP \\ \hline
0.1 & 2240.09 & 0.16 & 881.40 & 986.77       \\ 
0.2 & 12.25 & 0.01 & 6.86 & 6.61       \\ 
0.3 & 1.48 & $<0.01$ & 0.77 & 0.72       \\ 
0.4 & 0.21 & $<0.01$ & 0.1 & $<0.01$        \\ 
0.5 & 0.09 & $<0.01$ & $<0.01$ & $<0.01$        \\ \hline
\end{tabular}
\end{table}

\begin{table}[]
\centering
\caption{Running time for $\epsilon=0.2$ (in seconds)}\label{table:epsilon02s}
\begin{tabular}{crrrr}
\hline
$n$ & offline    & NOF        & CI          & BATSP \\ \hline
5  & 12.25 & 0.01 & 6.86 & 6.61 \\ 
10  & 302.63 & 0.18  & 65.33  & 66.26     \\ 
15  & 1785.35 & 0.61  & 117.83 &  117.89       \\ 
20  & 3673.34 & 0.83  & 154.92  & 189.88       \\ 
25  & 7890.86 & 1.88  & 413.26  & 424.49      \\ \hline
\end{tabular}
\end{table}

\subsubsection{Approximation ratio with different n} \label{sec:stable appro ratio}
As illustrated in Figure~\ref{fig:stable_appro} where $\epsilon$ is set to 0.2, the performance of our algorithms, except BATSP, is more or less stable with increasing $n$. The stability demonstrates the robustness of our algorithms that they can be used in various settings. BATSP algorithm produces an approximation ratio similar to the other algorithms when $n$ equals 5, but it increases to around 1.6 when $n$ is greater than 10 and further surges to 1.75 when $n$ is 25. This sharp increase demonstrates that the performance of BATSP is inferior to the other three algorithms.

\vspace{1\in}
\begin{figure}[h]
\centering
\includegraphics[width=0.48\textwidth]{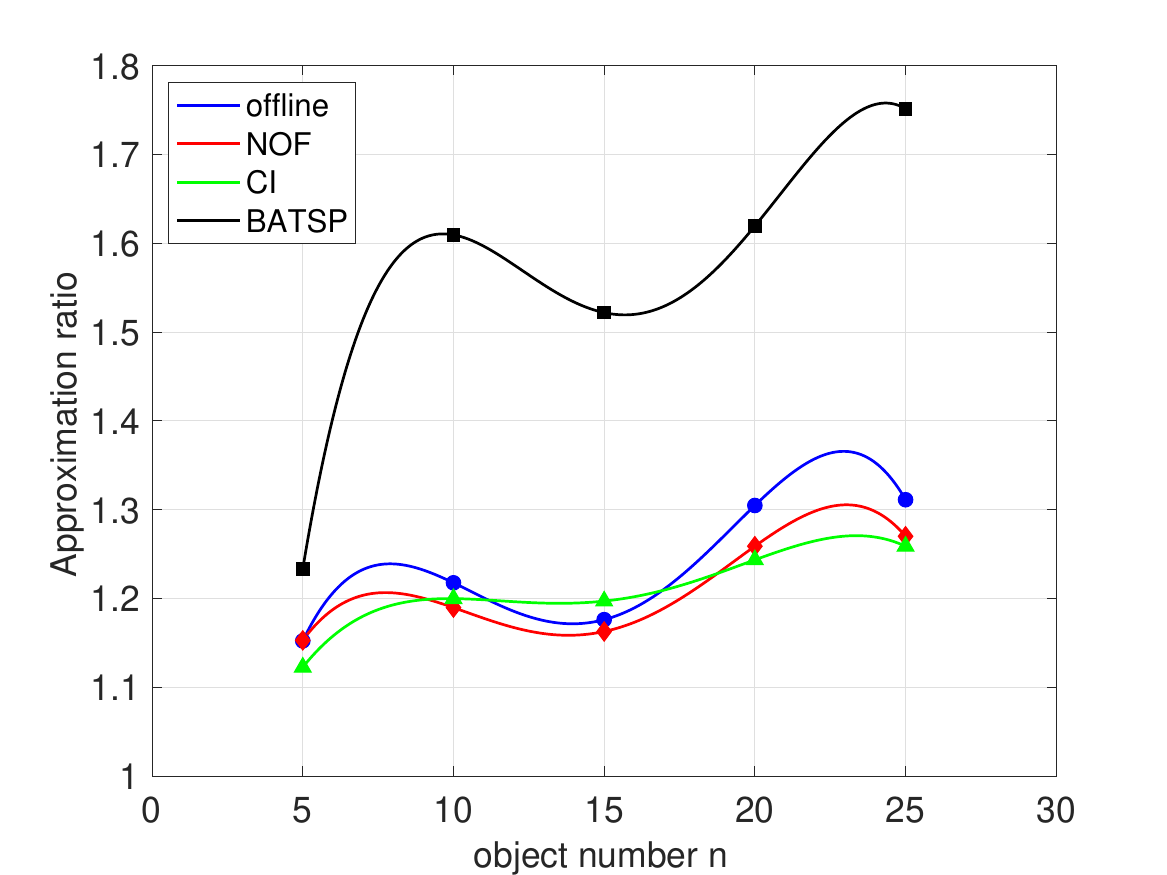}
\caption{Approximation ratio with $\epsilon=0.2$.}
\label{fig:stable_appro}
\vspace{-0.3in}
\end{figure}

\subsection{Discussion}
In general, there is a slightly rising trend in approximation ratio with the increase of $\epsilon$ when the object number is greater than 15, whereas in simple setting, the index is relatively stable. From Theorem~\ref{theorem1}, it is clear that a larger $\epsilon$ leads to a higher approximation ratio, since the error in the grid increases and there are fewer observation points that can be visited. We suppose that for simple cases, the effect of error in fewer visited observation points offsets that of Theorem~\ref{theorem1}.

The experiments also show that the offline method actually performs much better than its theoretical performance and NOF and CI achieve similar results, whereas from which BATSP's approximation ratio is apparently distinct. However, BATSP seeks to follow the best path whenever it makes a step forward. The reason could be that BATSP is based on the approximation algorithms during the whole process and there is approximation error for each step. Furthermore, the complex situation with 10 or more objects makes it more difficult for BATSP to plan its path via the TSP algorithm proposed by Christofides~\cite{1.5TSP}. Finally, it might be the accumulated error that leads to BATSP relatively unsatisfactory performance. 

Furthermore, it should be pointed out that the offline method, CI, and NOF achieve similar path cost with that of the Gurobi optimizer in much less time. When $\epsilon$ is 0.2, the proposed algorithms use within one second to several thousand seconds while the Gurobi optimizer requires 30000 seconds, which is the time limit defined earlier. Nevertheless, the path costs of these three methods are only 10 to 20 percent greater than the Gurobi best feasible results when $n$ is smaller than 25. Both the performance and the speed of our algorithms can be well demonstrated by these results, as shown in Figure~\ref{fig:tobestresult} and Figure~\ref{fig:stable_appro}.

\section{Conclusion}
In this work, we focus on planning a path with the lowest length for an UAV that aims to observe a set of objects. Each object is assumed to be a rectangle and each side should be observed while satisfying a quality constraint. We present an offline algorithm with an approximation ratio of $(1+\epsilon)(2+2n)$ and three online algorithms. This is the first work that presents path planning algorithms with an approximation ratio. Numerical results show that the offline algorithm and two online algorithms actually achieves an approximation ratio of around 2.1. Our algorithms can also yield comparable costs to an integer linear programming optimizer by around 1.2 times of it at a very fast speed, yet the latter uses 30000 seconds.

\end{document}